\documentclass[11pt]{article}
\usepackage{amsfonts,amssymb,amsmath,amsthm}
\usepackage[utf8]{inputenc}
\usepackage[T1]{fontenc}
\usepackage{newtxtext}
\usepackage[hidelinks]{hyperref}
\usepackage[shortlabels]{enumitem}
\usepackage{color}
\usepackage{tikz}
\usetikzlibrary{arrows,matrix,positioning}
\usepackage{tikz-cd}
\usepackage{graphicx,  fullpage, url}
\usepackage{verbatim}
\usepackage{mathrsfs}
\usepackage{mathtools}
\usepackage{array}
\usepackage{thm-restate}
\usepackage{setspace} 
\usepackage{thmtools}
\usepackage{thm-restate}
\newcommand{\uni}{\mathrm{uni}}

\renewcommand{\epsilon}{\varepsilon}

\def\notes{1}
 \newcommand{\TODO}[1]{\ifnum\notes=1{{\sf\color{red} [Todo: #1]}}\fi}
    \newcommand{\nrz}[1]{\ifnum\notes=1{{\sf\color{blue} [Noga: #1]}}\fi}
        \newcommand{\sk}[1]{\ifnum\notes=1{{\sf\color{blue} [Swastik: #1]}}\fi}
            \newcommand{\ssa}[1]{\ifnum\notes=1{{\sf\color{blue} [Shubhangi: #1]}}\fi}

\theoremstyle{plain}
\newtheorem{theorem}{Theorem}[section]
\newtheorem{lemma}[theorem]{Lemma}
\newtheorem{definition}[theorem]{Definition}
\newtheorem{corollary}[theorem]{Corollary}
\newtheorem{claim}[theorem]{Claim}
\newtheorem{conjecture}[theorem]{Conjecture}

\newtheorem{fact}[theorem]{Fact}

\title{Simple Constructions of Unique Neighbor Expanders \\ from Error-correcting Codes}

\author{ Swastik Kopparty\thanks{Department of Mathematics and Department of Computer Science, University of Toronto. Email:\href{mailto:swastik.kopparty@utoronto.ca}{swastik.kopparty@utoronto.ca}. Supported by an NSERC discovery grant.}
\and Noga Ron-Zewi\thanks{Department of Computer Science,
    University of Haifa. Email:
    \href{mailto:noga@cs.haifa.ac.il}{noga@cs.haifa.ac.il}. Supported in part
    by ISF grant 735/20 and by the European Union (ERC, ECCC, 101076663). Views and opinions expressed are however those of the author(s) only and do not necessarily reflect those of the European Union or the European Research Council. Neither the European Union nor the granting authority can be held responsible for them.} \and Shubhangi Saraf\thanks{Department of Mathematics and Department of Computer Science, University of Toronto. Email:\href{mailto:shubhangi.saraf@utoronto.ca}{shubhangi.saraf@utoronto.ca}. Supported by an NSERC discovery grant.}
}

\date{}

\begin{document}
 
\maketitle

\begin{abstract}

In this note, we give very simple constructions of unique neighbor expander graphs starting from spectral or combinatorial expander graphs of mild expansion. These constructions and their analysis are simple variants of the constructions of LDPC error-correcting codes from expanders, given by
Sipser-Spielman~\cite{SS96} (and Tanner~\cite{Tanner81}), and their analysis. We also show how to obtain expanders with many unique neighbors using similar ideas.

There were many exciting results on this topic recently, starting with Asherov-Dinur~\cite{AD23} and  Hsieh-McKenzie-Mohanty-Paredes~\cite{HMMP23}, who gave a similar construction of unique neighbor expander graphs, but using more sophisticated ingredients (such as almost-Ramanujan graphs) and a more involved analysis. Subsequent beautiful works of Cohen-Roth-TaShma~\cite{CRT23} and Golowich~\cite{Golowich23} gave even stronger objects (lossless expanders), but also using sophisticated ingredients. 

The main contribution of this work is that we get much more elementary constructions of unique neighbor expanders and with a simpler analysis.

\end{abstract}

\section{Introduction}\label{sec:intro}

This paper is about bipartite expander graphs, and we begin by introducing some
relevant notions of expansion.

Let $G$ be a bipartite graph with left set $L$, right set $R$, and where
all vertices in $L$ have degree $d$.
$G$ is called a $(\delta, \alpha)$-{\em (combinatorial)  expander} for $\delta, \alpha >0$ if for any subset $S \subseteq L$ of size at most $\delta |L|$,
the set $\Gamma(S)$ of neighbors of $S$ has size at least $\alpha d |S|$ (this is a constant fraction of
the maximum possible: $d|S|$). $G$ is called a {\em lossless expander} if $\alpha = 1-\epsilon$ for a small $\epsilon >0$.
Since for a small $\epsilon >0$, $\Gamma(S)$ is almost as big as $d |S|$ (which is the number of edges incident on $S$), almost every vertex in $\Gamma(S)$ has a
{\em unique} neighbor in $S$. This brings us to the definition of unique neighbor expanders.

A bipartite graph $G$ as above is called a $\delta$-\emph{unique neighbor expander} if
for any subset $S \subseteq L$ of size at most $\delta |L|$, there is some vertex in $R$ which is adjacent
to exactly one element of $S$. A stronger variant of this definition is a $(\delta, \alpha)$-\emph{unique neighbor expander}, which
asks for the existence of  $\alpha d|S|$ vertices in $R$ that are each adjacent to exactly one element of $|S|$. 
As mentioned above,   lossless expanders are (strong) unique neighbor expanders. 
In particular, the following fact is well-known. 

\begin{fact}\label{fact:comb_to_UN}
Suppose that $G$ is a $(\delta,1- \epsilon)$-combinatorial expander. Then $G$ is also a $(\delta, 1-2\epsilon)$-unique neighbor expander.
\end{fact}

Note that the above fact is only meaningful when $\alpha: = 1- \epsilon > \frac 1 2$.

\subsection{History: Unique Neighbor Expanders and Codes}

Unique neighbor expanders are of significant interest in coding theory because they immediately lead to nice error-correcting codes \cite{SS96} (we will shortly elaborate on this connection). They also have applications in the construction of locally testable codes \cite{DSW06,BV09}. 
In addition to coding theory, unique neighbor expanders have found applications for routing problems in online settings and for load balancing problems \cite{PU89,  ALM96, Pip96, AC02}. See~\cite{AD23} for a discussion of several nice applications of unique neighbor expanders. For now we will focus on the application of unique neighbor expanders to codes, since it is precisely the connection to the construction of codes using expanders that inspired our construction of unique neighbor expanders. 

In their fundamental paper~\cite{SS96}, Sipser and Spielman
proposed several constructions of Low Density Parity Check (LDPC) codes 
based on sufficiently good expanders. 
We start by describing the simplest such construction from that paper, which we call {\bf SS1}.

\paragraph{SS1:} Take an \emph{unbalanced} bipartite graph $G$ with $n$ vertices on the left and
$\beta n$ vertices on the right for some $\beta \in (0,1)$. The left vertices $u$ will index codeword coordinates; we will write bits $x_u$ on these vertices. The right vertices $v$ will represent constraints on these coordinates; each right vertex $v$ will ask that the parity of all the bits $x_u$ written on the neighbors $u$
of $v$ should equal $0$. The resulting codewords $x \in \{0,1\}^n$ form a linear code 
$C[G]$, and we call $G$ the \emph{parity-check graph} of the resulting code $C[G]$.

Now if $G$ is a \emph{unique neighbor expander}, then we get that for any $x$ of low weight, there is some right vertex $v$  that is adjacent to exactly one nonzero coordinate of $x$, and thus the constraint of $v$ is violated. This implies that low weight strings $x$ cannot be codewords, and thus the code has good distance. Since $\beta < 1$, the code also has $\Omega(n)$ dimension, and thus also has good rate. Furthermore, if $G$ has bounded right degree, then the obtained codes are Low Density Parity Check (LDPC) codes, since they are defined by parity constraints involving few variables. 

\medskip

At that time, there were no explicit constructions of unique neighbor expanders, and thus the SS1 expander-to-code
transformation of Sipser and Spielman did not lead to explicit error-correcting codes. The construction of explicit unique neighbor expanders was an important open question left by their work.
Nevertheless, Sipser and Spielman gave a different construction (related to a construction of Tanner~\cite{Tanner81})
of codes from expander graphs, and it turned out that this construction required just moderate expansion, and could
thus be made explicit from what was known then. We describe this construction, which we will call {\bf SS2}, below.

\paragraph{SS2:} For this construction, we start with a bipartite  graph $G$
with $n$ vertices on the left and $\beta n$ vertices on the right (for some $\beta \in (0,1)$),
with left and right degrees $d$ and $d/\beta$, respectively, as well as a constant size linear code $C_0 \subseteq \{0,1\}^{d/\beta}$.
The codewords of the new code $C[G, C_0]$ are again formed out of bits $x_u$ written on the left vertices, but this time they have to obey the constraint that for every right vertex $v$, the 
string $(x_u)_{u \in \Gamma(v)}$ should be a codeword of $C_0$. Note that the SS1 construction  is just the special case where $C_0$ is the parity code (i.e., the codewords are all strings of even parity).

It was shown by Sipser and Spielman that if $G$ is a \emph{combinatorial expander} with only a moderate (non-lossless) expansion, and $C_0$ has sufficient distance in terms of this expansion, then the code $C[G, C_0]$ has a  good distance. The analysis is a generalization of the simple argument given above for the special case where $C_0$ is the parity code. Specifically, it can be shown that if $G$ is a combinatorial expander with moderate expansion, then for any $x$ of low weight, there is some right vertex $v$ that is adjacent to a \emph{small number} of nonzero coordinates of $x$, and thus the constraint of $v$ is violated, assuming that $C_0$ has a sufficiently large distance. 

In their paper, Sipser and Spielman also gave an alternate construction, in which $G$ is
the edge-vertex incidence graph\footnote{The \emph{edge-vertex incidence graph} of a graph $G=(V,E)$ is the bipartite graph $G'=(L\cup R, E')$, where $L=E$, $R=V$, and $(e,v) \in E'$ if $v$ is one of the endpoints of $e$ in $G$.} of a 
\emph{spectral expander} of mild expansion (i.e., a (non-bipartite) regular graph of  bounded second normalized eigenvalue). In this case, the above analysis follows using the expander mixing lemma.

\medskip

Several years later, the first explicit unique neighbor expanders that allowed an instantiation of the SS1 construction were discovered.
This was first done by Capalbo, Reingold, Vadhan and Wigderson~\cite{CRVW}, who constructed lossless expanders (which are also unique neighbor expanders by Fact \ref{fact:comb_to_UN}), with any constant imbalance $\beta  = |R|/|L|>0$, via a variant of the zig-zag product of~\cite{RVW00}.
At around the same time, Alon and Capalbo \cite{AC02} gave a much cleaner construction of unique neighbor expanders\footnote{They also gave explicit
constructions of the more difficult non-bipartite unique neighbor expanders.}, by composing a 44-regular Ramanujan graph\footnote{Ramanujan graphs are extremal graphs with the best possible spectral expansion.} with
a certain special (43, 21)-vertex bipartite graph. Because of the specialized nature of this construction, it could only give
unique neighbor expanders with imbalance $\beta= 21/22$.

Recently, interest in this topic was revived in a beautiful work of Asherov and Dinur \cite{AD23}, who abstracted out the essential features of
the Alon-Capalbo construction. Specifically, they showed how to compose an "outer" large bipartite Ramanujan graph, together with an "inner" constant sized unique neighbor expander (which can be found via brute force search),  
to get a large unique neighbor expander, with any constant imbalance $\beta >0$. 
The composition was made via the \emph{routed product}, that was already implicit in \cite{AC02}, and formally defined in \cite{AD23}. 
However, the analysis of \cite{AD23} was quite subtle, and required the full Ramanujan property (i.e, the smallest possible second normalized eigenvalue) of the outer graph.

An independent beautiful work of Hsieh-McKenzie-Mohanty-Paredes~\cite{HMMP23} gave strong unique neighbor expanders (where small sets have many unique neighbors), with expansion from both the left and the right, via  a different kind of composition called the \emph{line product}. Further, their construction just used nearly Ramanjuan graphs 
as an ingredient, and did not need the full Ramanujan property (the line product also required the outer graph to be a Cayley graph).
Subsequent work of Cohen-Roth-TaShma~\cite{CRT23} and Golowich~\cite{Golowich23} gave much stronger objects (lossless expanders), with a very elegant construction and analysis, while also using significantly more powerful pseudorandom objects as ingredients such as high-dimensional expanders.

\subsection{Our result: Unique Neighbor Expanders from Codes}

The main focus of this paper is to give {\it simple} and {\it elementary} constructions of unique neighbor expanders. 

The starting point for our construction is the routed product of \cite{AC02, AD23}. 
As pointed out in \cite{AD23}, the routed product has a simple coding-theoretic interpretation.
Specifically, the routed product 
 $G_{\mathrm{out}} \circ G_{\mathrm{in}}$  
 of an outer graph $G_{\mathrm{out}}$ with an inner graph
$G_{\mathrm{in}}$
 is simply the parity-check graph of the SS2 construction
 $C[G,C_0]$, where  $G = G_{\mathrm{out}}$ and $C_0 = C[G_{\mathrm{in}}]$. 
As noted above,
\cite{AD23} chose the inner graph $G_{\mathrm{in}}$  to be a constant-size unique neighbor expander, and the outer graph $G_{\mathrm{out}}$ to be a  large Ramanujan graph, to obtain a large unique neighbor expander $G_{\mathrm{out}} \circ G_{\mathrm{in}}$. Furthermore, the analysis of \cite{AD23} was surprisingly subtle, and required the full Ramanujan property out of $G_{\mathrm{out}}$.

 In this work, we observe that the outer graph $G_{\mathrm{out}}$ can in fact  be chosen exactly as in the SS2 construction, namely, either 
  the edge-vertex incidence graph of a spectral expander with mild expansion, or a combinatorial expander with moderate expansion. Under either of these choices, we have that $C_0 = C[G_{\mathrm{in}}]$ is a good code by the SS1 analysis, and consequently  $C[G_{\mathrm{out}},C_0]$ is a good code by the SS2 analysis. We further observe that exactly the same analysis as that of the SS2 construction in fact shows something stronger: That $G_{\mathrm{out}} \circ G_{\mathrm{in}}$, which is the parity-check graph of $C[G_{\mathrm{out}},C_0]$, is a unique neighbor expander (which in particular implies that $C[G_{\mathrm{out}},C_0]$ is a good code by the SS1 analysis).
  
  Thus the old question of Sipser and Spielman about explicitly constructing unique neighbor expanders, required to instantiate the SS1 construction, can be solved by exactly the same methods as the SS2 construction!
 We further show that a slight strengthening of the SS2 analysis gives unique neighbor expanders with many unique neighbors. 

In more detail, in our first construction, given in Section \ref{sec:spectral}, we 
use as an outer graph the edge-vertex incidence graph of a spectral expander with a sufficiently small constant second normalized eigenvalue. 
Such graphs can be explicitly constructed for example by taking graph powers of an infinite family of regular spectral expanders of some constant second normalized eigenvalue (for example those of  Gabber and Galil \cite{GG81}), or, alternatively, 
using (a simple version of) the zig-zag product \cite{RVW00} (see also the survey \cite{Vadhan-survey}, Section 4.3). We show that this choice of an outer graph gives an explicit construction of an arbitrarily unbalanced strong unique neighbor expander (see Theorem \ref{thm:spectral_UN}).

In our second construction, given in Section \ref{sec:comb}, we use as an outer graph a biregular $(\delta, \alpha)$-combinatorial expander for some constant $\delta, \alpha>0$, and of sufficiently small imbalance $|R|/|L|$. We show that this choice leads to an explicit construction of an arbitrarily unbalanced $(\delta,\alpha')$-unique neighbor expanders with \emph{roughly the same expansion} $\alpha' = (1-\epsilon) \alpha$, for an arbitrarily small constant $\epsilon >0$ (see Theorem \ref{thm:comb_UN}). This can be viewed as an extension of the simple transformation from combinatorial to unique neighbor expanders, given by Fact \ref{fact:comb_to_UN}, to the regime of $\alpha < \frac 12$. 

As pointed out already by Sipser and Spielman,
one can explicitly construct biregular $(\delta, \alpha)$-combinatorial expanders for some constant 
$\delta, \alpha >0$, and an arbitrarily small imbalance $|R|/|L|$, by taking the edge-vertex incidence graph of a spectral expander of mild expansion, in which case our second construction reduces to our first one (though the analysis is a bit different). 
Beyond that, we are only aware of more sophisticated constructions of lossless combinatorial expanders \cite{CRVW,CRT23,Golowich23}, which to the best of our knowledge, are not biregular.
Our construction motivates the search for other simple constructions of biregular combinatorial expanders, as these will also lead to simple constructions of unique neighbor expanders of roughly the same expansion.

\section{Preliminaries}\label{sec:prelim}

In this section, we first provide the formal definitions of the types of expanders that will shall consider in this work: spectral, combinatorial, and unique neighbor expanders, and state known constructions of such graphs that we shall use in our constructions.
Then we formally define the routed product that will be used for composing these expanders. See \cite{HLW06} for an extensive survey of expander graphs and their applications in theoretical computer science.

\subsection{Spectral expander}

We start with the definition of  a spectral expander, which is a (non-bipartite) regular graph with bounded second eigenvalue.

\begin{definition}[Spectral expander]
We say that a $d$-regular graph $G$ on $n$ vertices is a $\lambda$-\emph{(spectral) expander} if $\max\{|\lambda_i|, \lambda_i\neq \pm d \}  \leq \lambda$, where  $d=\lambda_1\geq\lambda_2\geq \ldots \geq \lambda_n$ are the eigenvalues of the
adjacency matrix of $G$.
\end{definition}

For a graph $G=(V,E)$, and subsets $S,T \subseteq V$, we let $E(S,T)$ denote the set of edges with one endpoint in $S$ and another endpoint in $T$, and we let $|E(S,T)|$ denote the number of these edges, with the edges in $(S\cap T) \times (S\cap T)$ counted twice. A commonly used property of spectral expanders (and the only property of spectral expanders we shall use) is given by the \emph{expander mixing lemma} of Alon and Chung \cite{AC06}. This lemma says that for a $d$-regular graph $G=(V,E)$, and for any pair of subsets $S,T \subseteq V$, $|E(S,T)|$ is roughly equal the expected number of edges between $S$ and $T$ in a random $d$-regular graph.

\begin{theorem}[Expander Mixing Lemma, \cite{AC06}] \label{thm:eml}
Suppose that $G = (V,E)$ is a $d$-regular graph that is a $\lambda$-spectral expander. 
Then for any pair of subsets $S, T \subseteq V$,
$$\left|E(S,T)- \frac{d}{n}|S||T| \right| \leq \lambda \sqrt{ |S||T| }.$$
\end{theorem}

For our construction, we shall only need explicit constructions of infinite families of $d$-regular graphs which are $(\delta d)$-expanders for an arbitrarily small constant $\delta >0$, and an arbitrarily large degree $d$. Such graphs can be obtained for example by taking graph powers of an infinite family of $d$-regular $(\delta d)$-spectral expanders for some fixed non-negative integer $d$ and fixed $\delta >0$ (for example those of  Gabber and Galil \cite{GG81}), or, alternatively, 
using (a simple version of) the zig-zag product \cite{RVW00} (see also the survey \cite{Vadhan-survey}, Section 4.3).

\begin{theorem}[Explicit spectral expanders]\label{thm:spectral_explicit}
For any $\delta >0$ and non-negative integer $d_0$, there exists a non-negative integer $d \geq d_0$ so that there exists an explicit construction of an infinite family $\{G_n\}_n$ of graphs, where $G_n$ is a $d$-regular graph on $n$ vertices that is a $(\delta d)$-spectral expander.
\end{theorem}

To turn the above (non-bipartite) spectral expanders into bipartite graphs, we shall consider the \emph{edge-vertex incidence graph}, defined as follows.

\begin{definition}[Edge-vertex incidence graph]
The \emph{edge-vertex incidence graph} of a graph $G=(V,E)$ is the bipartite graph $G'=(L\cup R, E')$, where $L=E$, $R=V$, and $(e,v) \in E'$ if $v$ is one of the endpoints of $e$ in $G$. 
\end{definition}

\subsection{Combinatorial expander}

A combinatorial expander is a graph, where any not too large subset has many neighbors. In this work, we will only consider the bipartite one-sided version of such expanders, where only left subsets expand. 

We say that a bipartite graph $G=(L \cup R, E)$ is $d$-\emph{left regular} if $\deg(v) = d$ for any $v \in L$, and we say that $G$ is ($d_1, d_2$)-\emph{regular} if $\deg(v)=d_1$ for any $v \in L$ and $\deg(v)=d_2$ for any $v\in R $. We say that $G$ is \emph{$\beta$-unbalanced} if $|R| \leq \beta |L|$. 
For a vertex $v \in L \cup R$, let $\Gamma_G(v)$ denote the set of vertices adjacent to $v$, and for a subset $S \subseteq L$, let $\Gamma_G(S)$ denote the set of vertices adjacent to some vertex in $S$, i.e., $\Gamma_G(S) = \bigcup_{v\in S} \Gamma_G(v)$.  We shall sometimes omit the subscript $G$
if the graph $G$ is clear from the context. 

\begin{definition}[Combinatorial expander]
We say that a $d$-left regular bipartite graph
 $G = (L \cup R, E)$  is a
$(\delta,\alpha)$-\emph{(combinatorial) expander} if  $|\Gamma(S)| \geq \alpha d |S|$ for any non-empty subset $S \subseteq L$ with $|S| < \delta |L|$. 
\end{definition}

For our construction, we shall only need explicit constructions of infinite families of $\beta$-unbalanced $(\delta, \alpha)$-expanders for some fixed constants $\delta, \alpha >0$, and an arbitrarily small constant $\beta>0$.   Such constructions can be obtained for example by taking the edge-vertex incidence graph of the spectral expanders given in Theorem \ref{thm:spectral_explicit} \cite{SS96}.

\subsection{Unique neighbor expanders}

A unique neighbor expander is a strengthening of a combinatorial expander, where any not too large subset has many \emph{unique}  neighbors. Once more, we will only consider the bipartite one-sided version of such expanders.

Let $G = (L \cup R, E)$ be a bipartite graph.
We say that a vertex $v \in R$ is a \emph{unique neighbor} of a subset $S \subseteq L$ if $v$ is adjacent to exactly one vertex in $S$, and for a subset $S \subseteq L$, we let $\Gamma^{\uni}_G(S)$ denote the subset of all vertices $v \in R$ which are unique neighbors of $S$.  Once more, we shall sometimes omit the subscript $G$
if the graph $G$ is clear from the context. 

\begin{definition}[Unique neighbor (UN) expander]
We say that a $d$-left regular
bipartite graph
$G = (L \cup R, E)$ is a $\delta$-\emph{unique neighbor} (UN) \emph{expander} if $\Gamma^{\uni}(S) \neq \emptyset$ for any non-empty subset $S \subseteq L$ with $|S| <\delta |L|$. 
We say that $G$ is a 
$(\delta,\alpha)$-\emph{unique neighbor} (UN) \emph{expander} if $|\Gamma^{\uni}(S)| \geq \alpha d |S|$ for any non-empty subset $S \subseteq L$ with $|S| < \delta |L|$. 
\end{definition}

It follows by definition, that any $(\delta, \alpha)$-UN expander is also a $(\delta,\alpha)$-combinatorial expander. Conversely, it is not hard to show that any $(\delta, 1-\epsilon)$-combinatorial expander is also a $(\delta, 1-2\epsilon)$-UN expander. Note however that this latter implication is
only meaningful when $\epsilon < \frac 1 2$.

For our construction, we shall need (possibly non-explicit) arbitrarily unbalanced unique neighbor expanders. Such graphs (even with expansion arbitrarily close to $1$) can be shown to exist using the probabilistic method 
(see e.g., \cite[Lemma 11.2.3]{GRS-coding-book} \footnote{\cite[Lemma 11.2.3]{GRS-coding-book} actually shows the existence of combinatorial expanders with expansion arbitrarily close to $1$, but by Fact \ref{fact:comb_to_UN}, this also implies the existence of unique neighbor expanders with expansion arbitrarily close to $1$.}).

\begin{lemma}[Non-explicit unique neighbor expanders]\label{lem:UN_random}
For any $\epsilon > 0$ and $\beta \in (0,1]$, there exist $\delta  >0$ and non-negative integers $d$ and $n_0$, so that for all $n \geq n_0$, there exists a $d$-left regular $\beta$-unbalanced bipartite graph $G_n$ with $n$ left vertices that is a $(\delta, 1-\epsilon)$-UN expander.
\end{lemma}

\subsection{The routed product}

In this work, we will obtain explicit constructions of unique neighbor expanders by combining "outer graphs" which are either the edge-vertex incidence graph of explicit spectral expanders, or explicit combinatorial expanders, with  "inner graphs" which are non-explicit unique neighbor expanders,
via the \emph{routed product}, defined as follows (this product first appeared implicitly in \cite{AC02}, and formally defined later in \cite{AD23}).

\begin{definition}[The routed product \cite{AC02, AD23}] Let $G=(L \cup R,E)$ be an \emph{outer}  $(d_1,d_2)$-regular  bipartite graph, and let $G'=(L' \cup R', E')$
 be an \emph{inner} $d'$-left regular
 bipartite graph with $L'=[d_2]$. Furthermore, for any $v \in R$, assume an ordering  on the $d_2$ edges of $E$ incident to $v$ in $G$. 
 
 The \emph{routed product}  $G \circ G'$  is the 
 bipartite graph  $G \circ G':=(L'' \cup R'', E'')$, where 
$L'' = L$, $R'' = R \times R'$, and $(u, (v,v')) \in E''$ 
if  there exists $i \in [d_2]$ so that $(u,v)$
is the $i$-th edge incident to $v$ in $G$ and $(i,v') \in E'$. 
\end{definition}

It follows by definition, that $G \circ G'$ is $(d_1 \cdot d')$-left regular, and that if $G'$ is $\beta'$-unbalanced, then $G \circ G'$ is 
$ \tilde \beta$-unbalanced for 
$$ \tilde \beta = \frac{|R| \cdot |R'|} {|L|} = \frac{|R| \cdot |R'|} {d_2 \cdot |R|/d_1 }= \frac{d_1 \cdot |R'|} { |L'|} =d_1 \cdot \beta'.$$

\section{Unique Neighbor Expanders from Spectral Expanders}\label{sec:spectral}

In this section, we will show how to obtain a 
 simple explicit construction of unique neighbor expanders by combining "outer graphs" which are the edge-vertex incidence graph of explicit spectral expanders with  "inner graphs" which are non-explicit unique neighbor expanders via the routed product. 
 
 To this end, we first show  in Section \ref{subsec:spectral_one_UN}, as a warmup, how to obtain explicit unique neighbor expanders that only guarantee the existence of a  \emph{single} unique neighbor for any not too large left subset, using the analysis of error-correcting codes that is implicit in \cite{SS96}. Then in Section \ref{subsec:spectral_more_UN}, we show a slight strengthening of the analysis that guarantees a \emph{constant fraction} of unique neighbors. Finally, in Section \ref{subsec:spectral_instant}, we instantiate our transformation with the explicit spectral expanders given by Theorem \ref{thm:spectral_explicit}, and the non-explicit unique neighbor expanders given by Lemma \ref{lem:UN_random}, to obtain explicit unique neighbor expanders.

\subsection{One unique neighbor}\label{subsec:spectral_one_UN}

The following lemma says that the routed product of an outer edge-vertex incidence graph of a spectral expander with an inner unique neighbor expander yields a unique neighbor expander, where both unique neighbor expanders guarantee the existence of only a \emph{single} unique neighbor for any not too large left subset.

\begin{lemma}\label{lem:spectral_one_UN_comp}
Suppose that $G= (L \cup R, E)$ is the edge-vertex incidence graph of a $d$-regular graph that is a $\lambda$-spectral expander, and that $G' = (L' \cup R', E')$ is a $d'$-left regular bipartite graph that is a
$\delta'$-UN expander with $L'=[d]$. Then $G\circ G'=(L'' \cup R'', E'')$ 
is a $( \delta'  (\delta' - \frac \lambda d))$-UN expander.
\end{lemma}

The following lemma is implicit in the analysis of error-correcting codes of \cite{SS96}. Its proof is an easy consequence of the expander mixing lemma (Theorem \ref{thm:eml}). For completeness, and since in the next section we shall need a slight strengthening of this lemma, we provide a full proof of this lemma below. 

\begin{lemma}\label{lem:spectral_one_UN}
Let $G=(V,E)$ be a $d$-regular graph that is a $\lambda$-spectral expander, and let $S \subseteq E$ be a subset of edges in $G$. Suppose that for any $v \in V$ that is incident to some edge in $S$, it holds that $v$ is incident to at least $\delta  d$ edges in $S$. Then $|S| \geq \delta (\delta -  \frac \lambda d) |E|$.
\end{lemma}

\begin{proof}
Let $U \subseteq V$ be the subset of vertices in $G$ which are incident to some edge in $S$. Then by assumption, each vertex of $U$ is incident to at least $\delta d$ other vertices of $U$, and consequently we have that  
$|E(U,U)| \geq \delta d |U|$ (recall that by definition, each edge with both endpoints in $U$ is counted twice in $|E(U,U)|$). 
On the other hand, by the expander mixing lemma (Theorem \ref{thm:eml}), we have that 
$|E(U,U)| \leq \frac {d} {n} |U|^2 + \lambda |U|$. Combining these two inequalities gives that $|U| \geq (\delta - \frac \lambda d) n$. But by assumption that each vertex in $U$ is incident to at least $\delta d$ edges in $S$, this implies in turn that 
$$|S| \geq \frac 1 2  \delta d |U| \geq \delta  \left(\delta - \frac \lambda d\right) \cdot \frac {dn} {2} =\delta \left(\delta - \frac \lambda d\right) |E| .$$
\end{proof}

By the definition of the edge-vertex incidence graph, the above lemma can be equivalently stated as follows.

\begin{lemma}[Equivalent statement of Lemma \ref{lem:spectral_one_UN}]\label{lem:spectral_one_UN'}
Let $G=(L \cup R,E)$ be the edge-vertex incidence graph of a $d$-regular graph that is a $\lambda$-spectral expander, and let $S \subseteq L$ be a subset of left vertices in $G$. Suppose that 
$|\Gamma(v) \cap S| \geq \delta  d$ for any $v \in \Gamma(S)$. Then $|S| \geq \delta (\delta -  \frac \lambda d) |L|$.
\end{lemma}

We now proceed to the proof of Lemma \ref{lem:spectral_one_UN_comp}, based on the above lemma.

\begin{proof}[Proof of Lemma \ref{lem:spectral_one_UN_comp}]
Let $S  \subseteq L'' = L$ be a subset of left vertices in $G \circ G'$ of size $|S| < \delta'  (\delta' -  \frac \lambda d) |L''| = \delta'  (\delta' -  \frac \lambda d) |L|$. Then by Lemma \ref{lem:spectral_one_UN'}, there must exist a vertex $v \in \Gamma_G(S)$ so that $|\Gamma_G(v) \cap S| < \delta'  d$. 

Let $I \subseteq [d]$ be the subset which contains all indices of edges in $\Gamma_G(v)$ whose endpoints are contained in $S$. Then  $I$ is a non-empty subset satisfying that  $|I| < \delta'  d = \delta'  |L'|$, and since $G'$ is a 
$\delta'$-UN expander, there exists a vertex $v' \in R'$ which is a unique neighbor of $I$ in $G'$.

We claim that $(v,v')$ is a unique neighbor of $S$ in $G \circ G'$. To see this, let $i \in I$ be the unique neighbor of $v'$ in $I$ in the graph $G'$, and let $s \in S$ be so that $(s,v)$ is the $i$-th edge in $\Gamma_G(v)$. Then by the definition of the routed product, we have that $(s,(v,v')) \in E''$. Now suppose that $t \in S \setminus \{s\}$, and that  $(t,v)$ is the $j$-th edge in $\Gamma_G(v)$. Then we have that $j \in I \setminus \{i\}$, and so
 by assumption that $i \in I$ is the unique neighbor of $v'$ in $I$ in the graph $G'$, we have that $(j,v') \notin E'$. Consequently, by the definition of the routed product, $(t,(v,v')) \notin E''$. So $(v,v')$ is a unique neighbor of $S$ in $G \circ G'$.
\end{proof}

\subsection{Constant fraction of unique neighbors}\label{subsec:spectral_more_UN}

The following lemma says that the routed product of an outer edge-vertex incidence graph of a spectral expander with an inner unique neighbor expander yields a unique neighbor expander, where now both unique neighbor expanders guarantee the existence of a \emph{constant fraction} of unique neighbors for any not too large left subset.

\begin{lemma}\label{lem:spectral_more_UN_comp}
Suppose that $G= (L \cup R, E)$ is the edge-vertex incidence graph of a $d$-regular graph that is a
$\lambda$-spectral expander, and that $G' = (L' \cup R', E')$ is a $d'$-left regular bipartite graph that is a
$(\delta', \alpha')$-UN expander with $L'=[d]$. Then for any $\gamma \in (0,1)$, $G\circ G'=(L'' \cup R'', E'')$ 
is a $(\gamma  \delta' (\gamma \delta' - \frac \lambda d),  \frac { (1-\gamma)  \alpha'} {d})$-UN expander.
\end{lemma}

To prove the above lemma we shall need the following quantitative version of Lemma \ref{lem:spectral_one_UN}.

\begin{lemma}\label{lem:spectral_more_UN}
Let $G=(V,E)$ be a $d$-regular graph that is a $\lambda$-spectral expander, let $S \subseteq E$ be a subset of edges in $G$, and let
$U \subseteq V$ be the subset of vertices in $G$ which are incident to some edge in $S$. Suppose that at least a $\gamma$-fraction of the vertices of $U$ are incident to at least $\delta  d$ edges in $S$. Then $|S| \geq \gamma \delta (\gamma\delta -  \frac \lambda d) |E|$.
\end{lemma}

\begin{proof}
By assumption,  at least a $\gamma$-fraction of the vertices of $U$ are incident to at least $\delta  d$ other vertices in $U$, and consequently we have that
$|E(U,U)| \geq \gamma \delta d |U|$. 
On the other hand, by the expander mixing lemma (Theorem \ref{thm:eml}), we have that 
$|E(U,U)| \leq \frac {d} {n} |U|^2 + \lambda |U|$. Combining these two inequalities gives that $|U| \geq (\gamma \delta - \frac \lambda d) n$. But by assumption that at least a $\gamma$-fraction of the vertices of $U$ are incident to at least $\delta  d$ edges in $S$, this implies in turn that 
$$|S| \geq \frac 1 2  \gamma \delta d |U| \geq \gamma \delta  \left(\gamma \delta - \frac \lambda d\right)  \frac {dn} {2} =\gamma \delta \left(\gamma \delta - \frac \lambda d\right) |E| .$$
\end{proof}

As before, by the definition of the edge-vertex incidence graph, the above lemma can be equivalently stated as follows.

\begin{lemma}[Equivalent statement of Lemma \ref{lem:spectral_more_UN}]\label{lem:spectral_more_UN'}
Let $G=(L \cup R,E)$ be the edge-vertex incidence graph of a $d$-regular graph that is a $\lambda$-spectral expander, and let $S \subseteq L$ be a subset of left vertices in $G$. Suppose that $|\Gamma(v) \cap S| \geq \delta  d$
for at least a $\gamma$-fraction of the vertices in $\Gamma(S)$.
 Then $|S| \geq \gamma \delta  (\gamma \delta -  \frac \lambda d) |L|$.
\end{lemma}

We now proceed to the proof of Lemma \ref{lem:spectral_more_UN_comp}, based on the above lemma.

\begin{proof}[Proof of Lemma \ref{lem:spectral_more_UN_comp}]
Let $S  \subseteq L'' = L$ be a subset of left vertices in $G \circ G'$ of size $|S| < \gamma\delta'  (\gamma \delta' -  \frac \lambda d) |L''| = \gamma \delta' (\gamma \delta' -  \frac \lambda d) |L|$. Then by Lemma \ref{lem:spectral_more_UN'}, 
 $|\Gamma_G(v) \cap S| < \delta'  d$ for at least a $(1-\gamma)$-fraction of the vertices in $\Gamma_G(S)$. 

Let $v \in \Gamma_G(S)$ be a vertex which satisfies that  $|\Gamma_G(v) \cap S| < \delta'  d$.
As in the proof of Lemma \ref{lem:spectral_one_UN_comp}, let $I \subseteq [d]$ be the subset which contains all indices of edges in $\Gamma_G(v)$ whose endpoints are contained in $S$. Then  $I$ is a non-empty subset satisfying that  $|I| < \delta'  d = \delta'  |L'|$, and since $G'$ is a $( \delta', \alpha')$-UN expander, at least $\alpha' d'$ vertices $v' \in R'$ are unique neighbors of $I$ in $G'$. As in the proof of Lemma \ref{lem:spectral_one_UN_comp}, by the definition of the routed product, this implies in turn that 
$(v,v')$ is a unique neighbor of $S$ in $G \circ G'$ for at least $\alpha'  d'$ vertices $v' \in R'$. 

Overall,  we get that the number of unique neighbors of $S$ in $G \circ G'$ is at least
$$|\Gamma^{\uni}_{G \circ G'}(S)| \geq
\alpha' d' (1-\gamma) |\Gamma_G(S)| \geq \alpha' d' (1-\gamma) \frac{2 |S|} {d}  = \frac{ (1-\gamma) \alpha'} {d} \cdot  (2 d')|S|,$$
where the bound $|\Gamma_G(S)| \geq \frac{2 |S|} {d} $ follows since $G$ is $(2,d)$-regular. The conclusion follows by recalling that by the definition of the routed product, $G \circ G'$ is $(2d')$-left regular.
\end{proof}

\subsection{Instantiation}\label{subsec:spectral_instant}

Applying Lemma \ref{lem:spectral_more_UN_comp} with the outer graph being the explicit spectral expander given by Theorem \ref{thm:spectral_explicit}, and the inner graph being the non-explicit unique neighbor expander given by Lemma  \ref{lem:UN_random}, gives Theorem \ref{thm:spectral_UN}, restated below, which gives an explicit construction of a unique neighbor expander.

\begin{theorem}\label{thm:spectral_UN}
For any  constant $ \beta \in (0,1]$, there exist $\delta, \alpha>0$ and
a non-negative integer $c$,
 so that there exists an explicit construction of an infinite family $\{H_n\}_n$ of graphs, where $H_n$ is a 
 $c$-left regular bipartite graph with at least $n$ left vertices
and $\beta n$ right vertices that is a $(\delta,  \alpha)$-unique neighbor expander.
\end{theorem}

\begin{proof}
Let $\delta'>0$ and $d',n'_0$ be the constant and non-negative integers guaranteed by Lemma  \ref{lem:UN_random} for $\epsilon' = \frac 1 2$ and $\beta'= \frac \beta 2$. Let $d\geq d_0$ be the non-negative integer guaranteed by Theorem \ref{thm:spectral_explicit} for $\frac { \delta'}4$ and $d_0=n'_0$. Let $G'_d$  be the 
$d'$-left regular $\frac \beta 2$-unbalanced bipartite graph  with $d$ left vertices that is a $(\delta',\frac 1 2)$-UN expander, guaranteed by  Lemma  \ref{lem:UN_random}. Note that  $G'_d$ can be found in constant time via brute force search. Let $\{G_n\}_n$ be the explicit infinite family of graphs, where
 each $G_n$ is a $d$-regular graph on $n$ vertices that is a $(\frac { \delta'} 4 \cdot  d)$-spectral expander, guaranteed by Theorem  \ref{thm:spectral_explicit}.

Let $H_n := I(G_n) \circ G'_d$, where $I(G_n)$ is the edge-vertex incidence graph of $G_n$.
Then by the definition of the routed product, $H_n$ is a 
$(2d')$-left regular $\beta$-unbalanced bipartite graph with $\frac{d n} {2}$ left vertices.
Moreover, applying Lemma \ref{lem:spectral_more_UN_comp} with $\gamma = \frac 1 2$ gives that $H_n$ is a 
$(\frac{(\delta')^2} {8}, \frac 1 {4d})$-UN expander. So the theorem holds with $c = 2d'$, $\delta = \frac{(\delta')^2} {8}$, and $\alpha = \frac 1 {4d}$, which are all constant depending only on $\beta$.
\end{proof}

\section{Unique Neighbor Expanders from Combinatorial Expanders}\label{sec:comb}

In this section we show that composing an outer explicit
 combinatorial expander with an inner non-explicit unique neighbor expander via the routed product gives 
  an explicit unique neighbor expander with \emph{roughly the same expansion} as the outer combinatorial expander.

As in the previous section, we first show in Section \ref{subsec:comb_one_UN}, as a warmup, how to obtain explicit unique neighbor expanders that only guarantee the existence of a single unique neighbor, using the analysis of 
 error-correcting codes that is implicit in \cite{SS96}. Then in Section \ref{subsec:comb_more_UN}, we show how to extend the analysis to obtain a constant fraction of unique neighbors.  Finally, in Section \ref{subsec:comb_instant}, we instantiate our transformation with the non-explicit unique neighbor expanders given by Lemma \ref{lem:UN_random} to obtain the final transformation.

\subsection{One unique neighbor}\label{subsec:comb_one_UN}

The following lemma says that the routed product of an outer combinatorial expander with an inner unique neighbor expander yields a unique neighbor expander, where both unique neighbor expanders guarantee the existence of only a \emph{single} unique neighbor for any not too large left subset. 

\begin{lemma}\label{lem:comb_one_UN_comp}
Suppose that $G= (L \cup R, E)$ is a $(d_1,d_2)$-regular bipartite graph that is a
$(\delta,\alpha)$-expander for $d_2 > \frac 1 \alpha$, and that $G' = (L' \cup R', E')$ is a $d'$-left regular bipartite graph that is a
$\frac{1} {\alpha d_2}$-UN expander with $L'=[d_2]$. Then $G\circ G'=(L'' \cup R'', E'')$ is a 
$\delta$-UN expander.
\end{lemma}

The following simple lemma is implicit in the analysis of error-correcting codes of \cite{SS96}. For completeness, and since in the next section we shall need a slight strengthening of this lemma, we provide a full proof of this lemma below. 

\begin{lemma}\label{lem:comb_one_UN}
Let $G=(L \cup R,E)$ be a $(d_1,d_2)$-regular bipartite graph that is a
$(\delta,\alpha)$-expander, and let $S \subseteq L$ be a subset of left vertices in $G$. Suppose that $|\Gamma(v) \cap S| > \frac 1 \alpha$ for any $v \in \Gamma(S)$. Then $|S| \geq \delta |L|$.
\end{lemma}

\begin{proof}
By assumption that $|\Gamma(v) \cap S| > \frac 1 \alpha$ for any $v \in \Gamma(S)$, we have that 
$|E(S, \Gamma(S))| > \frac 1 \alpha |\Gamma(S)|$. On the other hand, since $G$ is $(d_1,d_2)$-regular, we have that 
$|E(S, \Gamma(S))| = d_1|S|$. Combining these two inequalities gives that $|\Gamma(S)| < \alpha d_1 |S|$. By assumption that $G$ is a $(\delta,\alpha)$-expander, this implies in turn that $|S| \geq \delta |L|$. 
\end{proof}

We now proceed to the proof of Lemma \ref{lem:comb_one_UN_comp}, based on the above lemma.

\begin{proof}[Proof of Lemma \ref{lem:comb_one_UN_comp}]
Let $S  \subseteq L'' = L$ be a subset of left vertices in $G \circ G'$ of size $|S| < \delta |L''| =  \delta |L|$. Then by Lemma \ref{lem:comb_one_UN}, there must exist a vertex $v \in \Gamma_G(S)$ so that $|\Gamma_G(v) \cap S| \leq \frac 1 \alpha$. 

Let $I \subseteq [d_2]$ be the subset which contains all indices of edges in $\Gamma_G(v)$ whose endpoints are contained in $S$. Then  $I$ is a non-empty subset satisfying that  $|I| < \frac 1 \alpha  = \frac {1 } {\alpha d_2} \cdot |L'|$, and since $G'$ is a $(\frac 1  {\alpha d_2})$-UN expander, there exists a vertex $v' \in R'$ which is a unique neighbor of $I$ in $G'$. By the definition of the routed product, this implies in turn that $(v,v')$ is a unique neighbor of $S$ in $G \circ G'$. 
\end{proof}

\subsection{Constant fraction of unique neighbors}\label{subsec:comb_more_UN}

The following lemma says that the routed product of an outer combinatorial expander with an inner unique neighbor expander yields a unique neighbor expander, where now both unique neighbor expanders guarantee the existence of a \emph{constant fraction} of unique neighbors for any not too large left subset.

\begin{lemma}\label{lem:comb_more_UN_comp}
The following holds for any $\gamma \in (0,1)$. 
Suppose that $G= (L \cup R, E)$ is a $(d_1,d_2)$-regular bipartite graph that is a
$(\delta,\alpha)$-expander for $d_2 > \frac 1 {\gamma\alpha}$, and that $G' = (L' \cup R', E')$ is a $d'$-left regular bipartite graph that is a
$(\frac{1} {\gamma \alpha d_2}, \alpha')$-UN expander with $L'=[d_2]$. Then $G\circ G'=(L'' \cup R'', E'')$ is a 
$( \delta, (1-\gamma) \alpha'  \alpha)$-UN expander.
\end{lemma}

To prove the above lemma we shall need the following  quantitative version of Lemma \ref{lem:comb_one_UN}.

\begin{lemma}\label{lem:comb_more_UN}
Let $G=(L \cup R,E)$ be a $(d_1,d_2)$-regular bipartite graph that is 
 a $(\delta,\alpha)$-expander, and let $S \subseteq L$ be a subset of left vertices in $G$. Suppose that $|\Gamma(v) \cap S| > \frac 1 {\gamma \alpha}$ for at least a $\gamma$-fraction of the vertices  in $ \Gamma(S)$. Then $|S| \geq \delta |L|$.
\end{lemma}

\begin{proof}
By assumption that $|\Gamma(v) \cap S| > \frac 1 {\gamma \alpha}$ for at least a $\gamma$-fraction of the vertices  in $ \Gamma(S)$, we have that 
$|E(S, \Gamma(S))| > \frac 1 \alpha |\Gamma(S)|$. On the other hand, since $G$ is $(d_1, d_2)$-regular, we have that 
$|E(S, \Gamma(S))| = d_1|S|$. Combining these two inequalities gives that $|\Gamma(S)| < \alpha d_1 |S|$. By assumption that $G$ is a $(\delta,\alpha)$-expander, this implies in turn that $|S| \geq \delta |L|$. 
\end{proof}

We now proceed to the proof of Lemma \ref{lem:comb_more_UN_comp}, based on the above lemma.

\begin{proof}[Proof of Lemma \ref{lem:comb_more_UN_comp}]
Let $S  \subseteq L'' = L$ be a subset of left vertices in $G \circ G'$ of size $|S| < \delta |L''| = \delta |L|$. Then by Lemma \ref{lem:comb_more_UN}, 
 $|\Gamma_G(v) \cap S| < \frac 1 {\gamma \alpha}$ for at least a $(1-\gamma)$-fraction of the vertices in $\Gamma_G(S)$. 

Let $v \in \Gamma_G(S)$ be a vertex which satisfies that  $|\Gamma_G(v) \cap S| < \frac {1} {\gamma \alpha}$.
As in the proof of Lemma \ref{lem:comb_one_UN_comp}, let $I \subseteq [d_2]$ denote the subset which contains all indices of edges in $\Gamma_G(v)$ whose endpoints are contained in $S$. Then  $I$ is a non-empty subset satisfying that  $|I| < \frac 1 {\gamma \alpha} = \frac 1 {\gamma \alpha d_2} \cdot |L'|$, and since $G'$ is a $(\frac 1 {\gamma \alpha d_2}, \alpha')$-UN expander, at least $\alpha' d'$ vertices $v' \in R'$ are unique neighbors of $I$ in $G'$. By the definition of the routed product, this implies in turn that 
$(v,v')$ is a unique neighbor of $S$ in $G \circ G'$ for at least $\alpha'  d'$ vertices $v' \in R'$. 

Overall, we get that the number of unique neighbors of $S$ in $G \circ G'$ is at least
$$|\Gamma^{\uni}_{G \circ G'}(S)| \geq
\alpha' d' (1-\gamma) |\Gamma_G(S)| \geq \alpha' d' (1-\gamma) \cdot (\alpha d_1 |S|) = (1-\gamma) \alpha' \alpha \cdot  (d_1 \cdot d') |S|,$$
where the bound $|\Gamma_G(S)| \geq \alpha d_1 |S| $ follows since $G$ is a $(\delta,\alpha)$-expander. The conclusion follows by recalling that by the definition of the routed product, $G \circ G'$ is $(d_1 \cdot d')$-left regular.
\end{proof}

\subsection{Instantiation}\label{subsec:comb_instant}

Applying Lemma \ref{lem:comb_more_UN_comp} with the inner graph being the non-explicit unique neighbor expander given by Lemma  \ref{lem:UN_random}, gives Theorem \ref{thm:comb_UN}, restated below, which shows how to transform combinatorial expanders into unique neighbor expanders with roughly the same expansion.

\begin{theorem}\label{thm:comb_UN}
For any constant $\epsilon, \alpha >0$, $\beta \in (0,1]$, and non-negative integer $d$, there exist $\mu >0$ and
a non-negative integer $c$ so that the following holds.
Suppose that there exists an explicit construction of an infinite family of graphs $\{G_n\}_n$, where each $G_n$ is 
a $(d,d/\mu )$-regular bipartite graph with $n$ left vertices 
 that is a $(\delta,\alpha)$-combinatorial expander.
Then there exists an explicit construction of an infinite family of graphs $\{H_n\}_n$, where each $H_n$ is a $c$-left regular 
graph with $n$ left vertices and $\beta n$ right vertices that is a $( \delta, (1-\epsilon)\alpha)$-unique neighbor expander.
\end{theorem}

\begin{proof}
Let $\delta'>0$ and $d', n'_0 $ be the constant and non-negative integers guaranteed by Lemma  \ref{lem:UN_random} for $\epsilon' := \frac \epsilon 2$ and $\beta' := \frac \beta d $. Let 
$\mu:=\min\{\frac{ \epsilon \alpha \delta' d} {2} ,\frac 1 {n'_0}\}$.
Let $G'_{d/\mu}$  be the $d'$-left regular $\frac \beta d$-unbalanced 
bipartite graph  with $d/\mu$ left vertices that is a 
$(\delta',1- \frac \epsilon 2)$-UN expander, guaranteed by  Lemma  \ref{lem:UN_random}. Note that  $G'_d$ can be found in constant time via brute force search.

Let $H_n := G_n \circ G'_{d/\mu}$.
Then by the definition of the routed product, $H_n$ is a $(d \cdot d')$-left regular
$\beta$-unbalanced 
bipartite graph with $n$ left vertices. Let $\gamma: = \frac \epsilon 2$, and note that by our choice of $\mu$ we have that
 $\frac{1} {\gamma \alpha (d/\mu)}  \leq \delta'$.
Therefore, applying Lemma \ref{lem:comb_more_UN_comp} with this choice of $\gamma$ gives that $H_n$ is a 
$(\delta, (1-\epsilon)\alpha)$-UN expander. 
\end{proof}

\paragraph{Acknowledgement} We thank Ronen Shaltiel for a discussion about existing constructions of bipartite expander graphs.

\bibliographystyle{alpha} 
\bibliography{bibfile}

\end{document}